\documentclass[technote, 9pt, twoside]{IEEEtran}   
\usepackage{cite} 
  
\ifCLASSINFOpdf 
  \usepackage[pdftex]{graphicx}
\else
  \usepackage[dvips]{graphicx}  
\fi
 
\usepackage[cmex10]{amsmath} 
\usepackage{amsfonts} 
\usepackage{amsthm}
\usepackage{mathtools}   
\usepackage{algorithmic}  
\usepackage[hidelinks]{hyperref} 
\usepackage{array}
\usepackage{stfloats} 
\usepackage{subcaption} 
\usepackage{url}
\usepackage{tikz}
\usetikzlibrary{arrows, decorations.pathmorphing, backgrounds, positioning, fit, shapes.geometric} 

\newcommand{\pos}{y} 
\newcommand{\inp}{r} 
\newcommand{\totalInput}{e}

\newcommand{\dreference}{d_{\mathrm{ref}}} 

\newcommand{\distPos}{\Delta} 

\newcommand{\contNumCoef}{q}
\newcommand{\contDenCoef}{p}
\newcommand{\vehNumCoef}{b}
\newcommand{\vehDenCoef}{a}
\newcommand{\numInteg}{\eta} 

\newcommand{\vehicleTf}{G}
\newcommand{\controllerTf}{R}

\newcommand{\openLoop}{M}
\newcommand{\openLoopPart}{\overline{\openLoop}}

\newcommand{\wb}{\epsilon}
\newcommand{\wbmax}{\wb_{\max}}

\newcommand{\numVeh}{N}

\newcommand{\lapl}{{L}}
\newcommand{\redLapl}{{\lapl_\mathrm{r}}}
\newcommand{\redLaplCO}{\hat{\lapl}}

\newcommand{\spatEig}{\lambda} 

\newcommand{\spatEigMin}{\spatEig_{\min}}
\newcommand{\spatEigMax}{\spatEig_{\max}}

\newcommand{\spatEigZ}{\gamma}

\newcommand{\contNode}{C}
\newcommand{\obsvNode}{O}
\newcommand{\distance}{d}
\newcommand{\distanceCO}{\distance_{CO}}
\newcommand{\weight}{w}
\newcommand{\weightCO}{\weight_{CO}}

\newcommand{\diagTransBlockEig}{T}
\newcommand{\diagTransBlockEigZ}{Z}
\newcommand{\diagTransBlockEigZij}{Z_{ij}}
\newcommand{\diagTransBlockEigWn}{\diagTransBlockEig_{j}}

\newcommand{\diagTransBlockEigMin}{\diagTransBlockEig_{{\min}}}
\newcommand{\diagTransBlockEigMax}{\diagTransBlockEig_{{\max}}}

\newcommand{\diagTransBlockCO}{\diagTransBlockEig_{CO}}

\newcommand{\setZi}{\mathcal{J}}
\newcommand{\tranFunCo}{\diagTransBlockCO}
\newcommand{\tranFunDist}{\diagTransBlockEig_\distPos}

\newcommand{\hinfnorm}{\mathcal{H}_{\infty}}

\tikzstyle{genGraphNode} = [minimum size=3mm, thick, node
distance=5mm] 
\tikzset{
    >=stealth', bend angle=10 }
\tikzstyle{normalNode} = [genGraphNode, circle,draw=black,fill=black!10,thick]
\tikzstyle{controllingNode} = [genGraphNode,
circle,draw=black!20,fill=black!60,very thick, text=white, font=\bfseries] 
\tikzstyle{observingNode} =[genGraphNode, circle,draw=black!50,fill=white,thick]
\tikzstyle{placeholder}=[genGraphNode, circle]

\tikzstyle{block} = [draw, fill=white!20, rectangle, 
    minimum height=2em, minimum width=3em] 
\tikzstyle{gain} = [draw,regular polygon,
        regular polygon sides=3, minimum height=2em, minimum
width=1em, inner sep=0pt] 
\tikzstyle{gainLeft} = [gain, shape border rotate=-90]
\tikzstyle{gainRight} = [gain, shape border rotate=90]
\tikzstyle{gainDown} = [gain, shape border rotate=180]
\tikzstyle{gainUp} = [gain, shape border rotate=90]
\tikzstyle{sum} = [draw, fill=white!20, circle,inner sep=1pt]
\tikzstyle{input} = [coordinate]
\tikzstyle{output} = [coordinate]
\tikzstyle{pinstyle} = [pin edge={to-,thin,black}]
\tikzstyle{junction} = [draw, circle, minimum height=0.02em, fill=black, inner sep=0pt]

\newtheoremstyle{assump}
{1pt}
{1pt}
{}
{}
{}
{)}
{-1.8em}
{}

\theoremstyle{plain}
\newtheorem{theorem}{Theorem}
\newtheorem{lemma}{Lemma}
\newtheorem{corollary}{Corollary}
\theoremstyle{definition}
\newtheorem{definition}{Definition}
\newtheorem{remark}{Remark}
\theoremstyle{assump}

\begin{document}

\title{Scaling in bidirectional platoons with dynamic controllers and proportional asymmetry}
%

        \author{Ivo~Herman, Dan~Martinec,
        Zden\v{e}k~Hur\'{a}k and Michael Sebek
\thanks{The authors are with the Faculty of Electrical Engineering (the first two are doctoral students), Czech Technical University in Prague. E-mail: ivo.herman@fel.cvut.cz. Supported by Czech Science Foundation within GACR 13-06894S (I.~H.).  
}}

\markboth{IEEE Transactions on Automatic Control}%
{Herman \MakeLowercase{\textit{et al.}}: Scaling of transfer functions in platoons}

\maketitle
\begin{abstract}   
We consider platoons composed of identical vehicles with an asymmetric nearest-neighbor interaction. We restrict ourselves to intervehicular coupling realized with dynamic arbitrary-order onboard controllers such that the coupling to the immediately preceding vehicle is proportional to the coupling to the immediately following vehicle. Each vehicle is modeled using a transfer function and we impose no restriction on the order of the vehicle. The platoon is described by a transfer function in a convenient product form. We investigate how the H-infinity norm and the steady-state gain of the platoon scale with the number of vehicles. We conclude that if the open-loop transfer function of the vehicle contains two or more integrators and the Fiedler eigenvalue of the graph Laplacian is uniformly bounded from below, the norm scales exponentially with the growing distance in the graph. If there is just one integrator in the open loop, we give a condition under which the norm of the transfer function is bounded by its steady-state gain---the platoon is string-stable. Moreover, we argue that in this case it is always possible to design a controller for the extreme asymmetry---the predecessor following strategy.
\end{abstract} 
\begin{IEEEkeywords}
Vehicular platoon, string stability, asymmetric control, scaling, transfer functions.
\end{IEEEkeywords}
\section{Introduction}
Vehicular platoons are chains of automatic cars that are supposed to travel with tight spacing in a highway lane. They are expected to increase the safety and capacity of highways. A number of theoretical results are available in the literature, but experiments with short vehicular platoons were described too \cite{Milanes2014} (PATH project) or \cite{Coelingh2012} (SARTRE project). Majority of the practical results rely on intervehicular communication. The most commonly adopted approaches are \textit{Cooperative Adaptive Cruise Control} (CACC) \cite{Milanes2014, Naus2010}, \textit{leader following}\cite{Peters2013} and \textit{leader's velocity transmission} \cite{Hao2012, Barooah2009a}. However, the communication can be delayed, disturbed or even denied by an intruder. 

In the absence of intervehicular communication, the only available information
is the one measured by the onboard sensors, especially the intervehicular
distances. It turns out that certain properties of such platoons need not scale
well for a growing number of vehicles. Among the strategies, the
\emph{time-headway policy} is scalable \cite{Middleton2010} but the platoon's length grows with the speed of the
leader. Among fixed-distance approaches such as the  \textit{predecessor
following} and \textit{symmetric or asymmetric bidirectional control}, an
unpleasant phenomenon known as \textit{string instability} can occur. This means
that a disturbance affecting a given vehicle can be amplified as it propagates
along the platoon (string) of vehicles. For the predecessor following strategy,
string instability  occurs for an arbitrary model of a vehicle as long as there
are least two integrators in the open loop \cite{Seiler2004a}. If
measurements of the distance from both the immediately preceding and the
immediately following vehicles are available, we call
the corresponding control bidirectional. In this paper we are going to
revolve around the role of asymmetry of bidirectional coupling.

Recent works suggest that in a bidirectional platoon with second-order open-loop
dynamics, a good trade-off between the settling time and peaks in the transient
response can be achieved if the asymmetry of coupling is imposed differently on
the measured intervehicular distances and their first derivatives---relative velocities. However, these
results are only obtained by numerical simulations \cite{Hao2012c}
or the results are based on reasonable conjectures \cite{Cantos2014a}. Moreover,
they are valid only for particular system models---double integrators. No general knowledge is available
so far. 

In contrast, if the coupling assumes identical asymmetry for both the distances
and their first derivatives, a nonzero lower bound on the formation eigenvalues
can be achieved \cite{Hao2012}. This guarantees controllability
\cite{Barooah2009a} of the formation of an arbitrary size. On the other hand,
for a double integrator model, the $\hinfnorm$ norm of a particular transfer
function related to disturbance attenuation grows exponentially in the number of
vehicles \cite{Tangerman2012}. Later this bad scaling was attributed to the
presence of the uniform bound on eigenvalues if there are at least two
integrators in the open loop \cite{Herman2013b}. Hence, the uniform boundedness
of eigenvalues plausible from the perspective of faster transient response must
be paid for by very bad scaling in the frequency domain.

If symmetric coupling is implemented, the norm grows only linearly \cite{Veerman2007, Hao2012b} but the step response suffers from very long transients---the eigenvalues get arbitrarily close to the origin. This can be alleviated using a wave-absorbing controller implemented on either end of the platoon \cite{Martinec2014b}. 
Finally, it is also the sensitivity of the platoon to the noise that depends on the number of integrators in the open loop \cite{Bamieh2012}.

In this paper we consider platoons composed of identical vehicles with an
asymmetric nearest-neighbor interaction. We restrict ourselves to the
case when the coupling to the immediately preceding vehicle is proportional to
the coupling to the immediately following vehicle (see eq. (\ref{eq:regErr})).
Each vehicle is modeled by a transfer function and we impose no restriction on
the order or structure of the model. 

We investigate how the $\hinfnorm$ norm and the steady-state gain of the platoon
scale with the number of vehicles. If the vehicle contains two or more
integrators and the eigenvalues of the graph Laplacian are uniformly bounded
from below, the norm scales exponentially with the growing distance in the
graph (Sec. \ref{sec:exponentialScaling}). If there is just one integrator in
the open loop, we give a condition under which the norm of the transfer function is bounded by its steady-state
gain---the platoon is string-stable (Sec. \ref{sec:stringStabCont}). In addition,
in this case it is possible to design a string-stable controller for the extreme asymmetry---the predecessor
following strategy, which offers some implementation advantages compared to general
asymmetric bidirectional control (see Sec. \ref{sec:prefFolDesign}).

The novelty is that our results hold for an \emph{arbitrary LTI model} (order
and structure) of the individual vehicle. Thus, we do not limit ourselves to a
single or double integrator as in \cite{Barooah2009a, Hao2012, Tangerman2012,
Bamieh2012, Hao2012c}. In fact, our work generalizes those results to arbitrary
transfer function models of individual vehicles. The main distinguishing feature is
the number of integrators in the open loop. We extend the result on exponential scaling from our paper
\cite{Herman2013b} to an arbitrary transfer function in the formation. Moreover, we add a discussion of scaling when only one integrator in the open loop is present in the agent model and also a steady-state gain is analyzed. This paper therefore
should give a broader qualitative overview of what is achievable with proportional
asymmetry for general vehicle models.


\section{Vehicle and platoon modelling}
Consider $\numVeh$ \emph{identical} vehicles indexed as $i=1, 2, \ldots, \numVeh$, with $i=1$ corresponding to the platoon leader. The leader drives independently of the platoon. The vehicles have identical transfer functions $\vehicleTf(s)=\frac{\vehNumCoef(s)}{\vehDenCoef(s)}$  of an arbitrary type and order with positions $\pos_i$ as the outputs. The input to the vehicle is produced by a dynamic controller $\controllerTf(s)=\frac{\contNumCoef(s)}{\contDenCoef(s)}$.
The open-loop model  $\openLoop(s)=\controllerTf(s) \vehicleTf(s) = \frac{\vehNumCoef(s)\contNumCoef(s)}{\vehDenCoef(s)\contDenCoef(s)}$ is a series connection of the controller and the vehicle models. 

\begin{definition}[Number of integrators in the open loop] 
Let the open-loop model be factored as
$\openLoop(s)=1/s^\numInteg\,\, \openLoopPart(s)$ with $\openLoopPart(0) < \infty$.
Then $\numInteg \in \mathbb{N}_0$ is the number of
integrators in the open loop.  
\end{definition}
The number $\numInteg$ is also known as \emph{type number} of the system. For instance, 
the model $\openLoop(s)=\frac{1}{s(s+a)}$ is a system
with one integrator in the open loop and $\openLoop(s)=\frac{s+1}{s^2
(s+b)}$ has $\numInteg=2$.
We call the well-known cases with $\openLoopPart(s)=1$ a
single-integrator system for $\numInteg=1$ and a double-integrator system for
$\numInteg=2$, respectively.

The input to the controller is the combined front and rear intervehicular spacing error
\begin{equation}
	\totalInput_i = (\pos_{i-1}-\pos_i) - \wb_i
	(\pos_{i}-\pos_{i+1}) + \inp_i. \label{eq:regErr}
\end{equation}
We call the nonnegative weight $\wb_i$ of the rear spacing error the 
\textit{constant of bidirectionality}. The general external input $\inp_i$
can represent, for instance, a measurement noise or a reference
such as the reference distance ${\dreference}$. In such a case $\inp_i =
-\dreference + \wb_i \dreference$ and the distances $\Delta_i=\pos_{i-1}-\pos_{i}$ are regulated to $\dreference$. The leader's control input is just $\inp_1$ and the controller of the trailing vehicle has the input $\totalInput_\numVeh=(\pos_{\numVeh-1}-\pos_\numVeh) + \inp_\numVeh$. Since we use a dynamic controller, the control law can also access the relative velocity and other derivatives of the distances (for instance by using a PD controller $\controllerTf(s)=\alpha s + \beta$).

\subsection{Laplacian properties}
The regulation errors in (\ref{eq:regErr}) are given in a vector form as
$\totalInput = -\lapl \pos + \inp$ with $\totalInput=[\totalInput_1, \ldots,
\totalInput_\numVeh]^T, \pos=[\pos_1, \ldots, \pos_\numVeh]^T$ and $\inp=[\inp_1, \ldots, \inp_\numVeh]^T$. The matrix $\lapl=[l_{ij}] \in \mathbb{R}^{\numVeh \times \numVeh}$ is the Laplacian of a path graph and has the following structure
{\small
\begin{equation}
	\lapl = \left[  
				\begin{matrix}
					0 & 0 & 0 & 0 & \ldots \\
					-1  & 1+\wb_2 & -\wb_2 & 0 & \ldots \\
					\vdots  & \vdots & \vdots & \ddots & \vdots \\
					 0 & \ldots & -1  & 1+\wb_{\numVeh-1} & -\wb_{\numVeh-1}\\
					0 & \ldots  & 0&  -1 & 1
				\end{matrix} 
			\right]
			\label{eq:laplacianLeader}
\end{equation}
}It is a non-symmetric tridiagonal matrix. Next we state some useful properties of $\lapl$, mainly taken from the literature.
\begin{lemma}
Laplacian $\lapl$ in (\ref{eq:laplacianLeader}) and its eigenvalues $\spatEig_i$ have the following properties:
	\begin{itemize}
	  \item[a)] The eigenvalues $\spatEig_i$ are all real and $\spatEig_i \geq 0,
	  \,\, \forall i$.
	  \item[b)] With the eigenvalues ordered as $\spatEig_1 \leq \spatEig_2 \leq \ldots \leq \spatEig_\numVeh$, the smallest eigenvalue $\spatEig_1=0$ and this eigenvalue is simple.
	  \item[c)] The eigenvalues are upper-bounded by $\spatEigMax$, that is, $\spatEig_i \leq \spatEigMax \leq 2 \max (l_{ii})$.
	  \item[d)] Let $\redLapl$ be the matrix obtained from $\lapl$ by deleting the first row and the first column (both correspond to the leader). Then $\spatEig_i(\lapl)=\spatEig_i(\redLapl)$ for all $\spatEig_i \neq 0$.
	  \item[e)] Suppose that $\wb_i \leq \wbmax < 1 \, \forall \, i$. Then the nonzero eigenvalues $\spatEig_2, \ldots, \spatEig_\numVeh$ are upper-bounded by $\spatEig_i \leq \spatEigMax = 2(1+\wbmax), \forall i \geq 1$ and lower-bounded by
	\begin{equation}
		\spatEig_i \geq \spatEigMin \geq
		\frac{1}{2}\frac{(1-\epsilon_{\max})^2}{1+\epsilon_{\max}}>0, \qquad\forall i \geq 2.
		\label{eq:uniformBound}
	\end{equation}
	The bounds are uniform, that is, they do not depend on $\numVeh$. 
	\item[f)] Let $\lapl_k$ be a matrix obtained from $\lapl$ by deleting $k$th row and column. Let the eigenvalues of $\lapl_k$, $1 < k < n$, be $\mu_1 < \mu_2 < \ldots < \mu_{n-1}$. Then
	\begin{equation}
		\lambda_{j+2} \geq \mu_{j} \geq \lambda_{j},  \, j=1,2, \ldots, N-2.
		\label{eq:eigenvalueInterlacing}
	\end{equation}
	\end{itemize}
	\label{lem:laplProp}
\end{lemma}
\begin{proof}
The properties a)-d) are discussed in \cite[Lem. 1]{Herman2013b}, e) is proved
in \cite[Thm. 1]{Herman2013b}. The statement f) follows from \cite[Thm.
5.5.6]{Fallat2011}, which gives conditions of interlacing for totally
nonnegative matrices. $\lapl$ is similar to a totally nonnegative matrix
\cite[pp. 6,7]{Fallat2011}. Both $\lapl$ and $\lapl_k$ can be
transformed to totally nonnegative matrices using similarity transform with
signature matrices $S=\text{diag}[1, -1, \ldots, 1, -1]$. The results are
$|\lapl|$ and $|\lapl_k|$ with the absolute values taken element-wise. Since
$|\lapl_k|$ is a principal submatrix of $|\lapl|$, interlacing occurs. Since
$\lapl$ is similar to $|\lapl|$ and $\lapl_k$ to $|\lapl_k|$, their
eigenvalues interlace.
\end{proof}

The property e) is an instance of uniform boundedness---the lower bound on
eigenvalues $\spatEigMin>0$ does not depend on $\numVeh$ \cite{Hao2012,
Tangerman2012, Herman2013b}. Applying f) repeatedly, the interlacing holds for
any principal submatrix. The eigenvalue $\spatEig_2$ is known as the Fiedler eigenvalue.
\begin{remark}
	 In \cite{Herman2013b} we considered a more general model with different controller weight $\mu_i$ for each vehicle such that $\lapl_\mu=W \lapl$, $W=\text{diag}[\mu_1, \mu_2, \ldots, \mu_\numVeh]$. For the clarity of presentation we restricted ourselves here to $\lapl$ in (\ref{eq:laplacianLeader}) and $\mu_i=1 \, \forall i$, although all the results (apart from the steady-state gain) would remain unchanged.
\end{remark}

\subsection{Transfer functions}
\begin{figure}
\centering
	\includegraphics[width=0.45\textwidth]{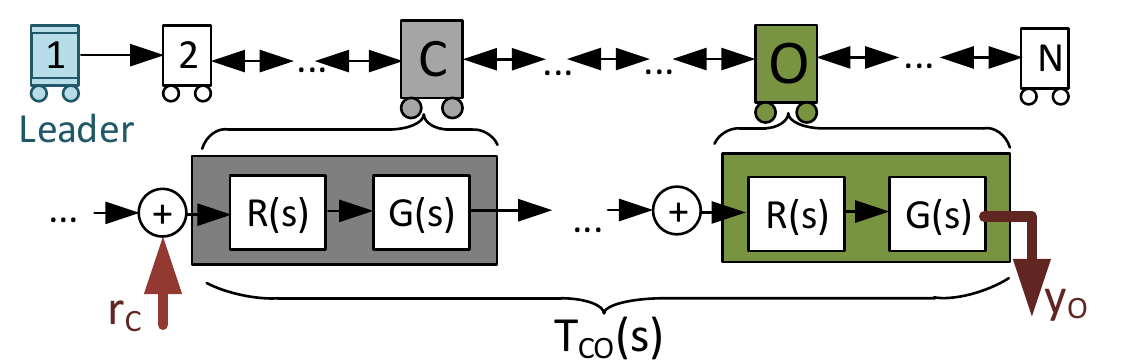}
	\caption{Block diagram showing the transfer function $\tranFunCo(s)$.}
	\label{fig:tranFun} 
\end{figure} 

We are interested in how the vector of external inputs $\inp$ (acting at the inputs of the controller) affects the vector of positions $\pos$ of vehicles. This is in general described by a transfer function matrix $y(s)=\mathbf{T}(s) r(s)$. The $(\obsvNode, \contNode)$th element of matrix $\mathbf{T}(s)$ is denoted by $\tranFunCo(s)=\frac{\pos_\obsvNode(s)}{\inp_\contNode(s)}$, $\contNode=1, \ldots, N, \obsvNode=1, \ldots, N$. The transfer function $\tranFunCo(s)$ therefore describes the effect of the external input $\inp_\contNode$ acting at a
vehicle indexed $\contNode$ (called a \emph{control vehicle}) on the
position $\pos_\obsvNode$ of the vehicle with an index $\obsvNode$ (called
an \emph{output vehicle})---see Fig. \ref{fig:tranFun}.
We will be interested in how its $\hinfnorm$ norm defined as $\|
\tranFunCo(s)\|_\infty=\sup_{\omega \geq 0} |\tranFunCo(\jmath \omega)|$ scales
with a growing number $\numVeh$ of vehicles and the distance $\distanceCO$ in a
graph. We use
the statement ''from $\contNode$ to $\obsvNode$'' with the meaning of ''from the
input $\inp_\contNode$ of the vehicle $\contNode$ to the output $\pos_\obsvNode$
of the vehicle $\obsvNode$''.
The indices $\contNode$ and $\obsvNode$ can be chosen arbitrarily. Note that due to bidirectional architecture, for any selection of $\contNode, \obsvNode$ the transfer function $\tranFunCo(s)$ depends on the whole formation.

Since the graph of a platoon is a path graph, there is only one directed
path from the node $\contNode$ to the node $\obsvNode$. This path is a sequence
of edges with the weights $\weight_{i,j}$. The weight of the path is $\weightCO
= \prod_{j=\contNode}^{\obsvNode-1} \weight_{j,j+1}$. In our case
$\weight_{i, i+1}=1$ and $\weight_{i+1, i}=\wb_i$, so
\begin{equation}
	\weightCO = \begin{cases}
		1 & \text{ for } \contNode \leq \obsvNode, \\
		\prod_{i=\obsvNode}^{\contNode-1} \wb_i & \text{ for } \contNode > \obsvNode.
	\end{cases}
	\label{eq:pathWeight}
\end{equation}
The number of edges on the directed path from the node $\contNode$ to the node $\obsvNode$ is called the graph distance $\distanceCO$ between $\contNode$ and $\obsvNode$. We use the following product form of $\tranFunCo(s)$ that we derived in \cite[Thm. 5]{Herman2014a} 
\begin{equation}
		\diagTransBlockCO(s)=\weightCO \frac{\left[ \vehNumCoef(s)\contNumCoef(s)
		\right]^{\distanceCO\!+\!1}
		{{\prod_{i=1}^{\numVeh\!-\!\distanceCO\!-\!1}}}{[\vehDenCoef(}s)
		\contDenCoef(s)\!+\!\spatEigZ_i \vehNumCoef(s)\contNumCoef(s)]
		}{\prod_{i=1}^{\numVeh} [\vehDenCoef(s)\contDenCoef(s) +
		\spatEig_j \vehNumCoef(s)\contNumCoef(s)]},
		\label{eq:tranFunGen} 
	\end{equation}  
where $\spatEig_j$ is the $j$th eigenvalue of $\lapl$. The coefficients
$\spatEigZ_i \in \mathbb{R}$, $\spatEigZ_i\leq \spatEigZ_{i+1}$, are the
eigenvalues of the matrix $\redLaplCO \in  \mathbb{R}^{\numVeh-\distanceCO-1
\times \numVeh-\distanceCO-1}$ that is obtained from $\lapl$ by deleting all the
rows and columns corresponding to the nodes on the path from $\contNode$ to
$\obsvNode$, see \cite[Thm. 10]{Herman2014a}. Note that $\redLaplCO$ is a
principal submatrix of $\lapl$, hence interlacing in the sense of Lemma
\ref{lem:laplProp} f) holds. For instance, for a formation with $\contNode=3$,
$\obsvNode=4$ and $\numVeh=5$, we delete the third and the fourth rows and
columns of $\lapl$ to get $\redLaplCO$ with the eigenvalues
$\gamma_i=[0, 1, 1+\wb_2]$,

{\small \arraycolsep=1.6pt
\begin{IEEEeqnarray}{rCl}
	\lapl &=& \begin{bmatrix}
					0 & 0 & 0 & 0 & 0 \\
					-1 & 1+\wb_2 & - \wb_2 & 0 & 0 \\ 
					0 & -1 & 1 + \wb_3 & -\wb_3 & 0 \\
					0 & 0 & -1 & 1+\wb_4 & -\wb_4 \\
					0 & 0 & 0 & -1 & 1
				\end{bmatrix} 
				\Longrightarrow 
				\redLaplCO = \begin{bmatrix}
					0 & 0 & 0 \\
					-1 & 1+\wb_2 &  0 \\ 
					0 & 0 & 1
				\end{bmatrix}.
\end{IEEEeqnarray}
\normalsize
}

Using the statement d) in Lemma \ref{lem:laplProp}, we can exclude the leader
from the formation (and also get rid of $\spatEig_1=0$ and $\spatEigZ_1=0$). Whenever we analyze a
transfer-function norm, we will work with $\redLapl$ and all the indices will
start from 2. The leader can be again included afterwards by multiplying the
transfer function $\tranFunCo(s)$ by $\openLoop(s)$.

\begin{remark}
Throughout the paper we assume that the overall system is asymptotically stable
for all $\numVeh$. It follows from (\ref{eq:tranFunGen}) that the polynomial
$\vehDenCoef(s)\contDenCoef(s) + \spatEig_j \vehNumCoef(s)\contNumCoef(s)$ must
be stable for any $\spatEig_j \in [\spatEigMin, \spatEigMax]$, $\spatEig_j \in
\mathbb{R}$ (similarly to \cite{Fax2004a}). Note that $\vehDenCoef(s)\contDenCoef(s)+\spatEig_j
\vehNumCoef(s)\contNumCoef(s)$ is a standard form for the denominator in the root-locus theory for the system $\spatEig_j \openLoop(s)$ with the gain $\spatEig_j$. Thus, we just need to stabilize the single-agent system $\spatEig_j \openLoop(s)$ for a bounded interval of the real gain $\spatEig_j \in [\spatEigMin, \spatEigMax]$. If $\spatEigMin > 0$, we can stabilize even a formation of unstable agents. From (\ref{eq:eigenvalueInterlacing}) it follows that also
$\spatEigZ_i \in [\spatEigMin, \spatEigMax], \, \forall i$, so if the system is asymptotically stable, all its zeros are in the left half-plane too.
\end{remark}

\section{Steady-state gain of transfer functions}
Besides the $\mathcal{H}_\infty$ norm, another important control-related
characteristic of a platoon is the steady-state gain $\tranFunCo(0)$. By the internal model principle \cite{Wieland2011} we assume
that $\numInteg\geq 1$ to enable the vehicles to track the leader's constant velocity.
With at least one integrator in $\openLoop(s)$ we get
$\vehDenCoef(0)\contDenCoef(0)=0$. After excluding the leader, the steady-state
gain follows from (\ref{eq:tranFunGen})
as
	\begin{IEEEeqnarray}{rCl}
	\diagTransBlockCO(0) &=&
	\weightCO \frac{\left[ \vehNumCoef(0)\contNumCoef(0) 
		\right]^{\distanceCO+1}
		{\prod_{i=2}^{\numVeh-\distanceCO-1}}{[\spatEigZ_i}
		\vehNumCoef(0)\contNumCoef(0)] }{{\prod_{j=2}^{\numVeh}} [
		\spatEig_j \vehNumCoef(0)\contNumCoef(0)]}
\nonumber \\
&=& \weightCO \frac{\prod_{i=2}^{\numVeh-\distanceCO-1}
	\spatEigZ_i}{\prod_{j=2}^{\numVeh} \spatEig_j}.
	\label{eq:dcGain}
	\end{IEEEeqnarray}

This shows that the steady-state gain \emph{does not depend on the dynamic model} of an
individual agent, it is only a function of the structure of the network
($\spatEig_j$ and $\spatEigZ_i$ are both obtained from $\lapl$). We can now
apply the previous result to get the steady-state gain of the transfer function
$\diagTransBlockCO(s)$ in vehicular platoons.

\begin{theorem} 
The steady-state gain of the platoon is given by
\begin{equation}
	\diagTransBlockCO(0) = 
	\begin{cases}
			\weightCO
\left(1+\sum_{i=1}^{\contNode-2} \prod_{j=1}^{i} \wb_{\contNode-j} \right) &
\text{ for } \contNode \leq \obsvNode \\
 	\weightCO \left(1+\sum_{i=1}^{\obsvNode-2}
		\prod_{j=1}^{i} \wb_{\obsvNode-j} \right) & \text{ for } \obsvNode < 
		\contNode
	\end{cases}
	\label{eq:dcgainplatoon}
\end{equation}
\label{thm:dcgain}
\end{theorem}  
The proof is in Appendix \ref{sec:pfDcGain}. Note that for $\contNode \leq
\obsvNode$, the steady-state gain does not depend on $\obsvNode$ as
$\weightCO=1$ for $\contNode \leq \obsvNode$.
We can discuss several cases relevant for the platoon control. 

 \begin{corollary}
	If there is  $\wb_{\max}$ such that
	$\wb_i \leq \wb_{\max} < 1 \, \,\forall i$, then $\tranFunCo(0)$ is upper bounded as $\diagTransBlockCO(0) \leq \frac{1}{1-\wb_{\max}}$. This holds for all $\numVeh$ and for all $\contNode, \obsvNode$.
\end{corollary}

\begin{proof}
We can bound the product in (\ref{eq:dcgainplatoon}) as $\prod_{j=1}^{i}
\wb_{\contNode-j} \leq \wb_{\max}^i$. Then $\diagTransBlockCO(0) \leq \weightCO \left(1+\sum_{i=1}^{\contNode-2} \wb_{\max}^i \right) \leq \weightCO \frac{1}{1-\wb_{\max}}$, since $\sum_{i=0}^\infty \wb_{\max}^i = \frac{1}{1-\wb_{\max}}$. The same holds for $\weightCO \left(1+\sum_{i=1}^{\obsvNode-2} \prod_{j=1}^{i} \wb_{\obsvNode-j} \right) \leq \weightCO \frac{1}{1-\wb_{\max}} $. If $\contNode \leq \obsvNode$, then $\weightCO=1$. If $\contNode > \obsvNode$, then $\weightCO = \prod_{i=\contNode-1}^{\obsvNode} \wb_i \leq \wbmax^{\distanceCO} < 1$. Therefore, $\diagTransBlockCO(0) \leq \weightCO \frac{1}{1-\wb_{\max}} \leq \frac{1}{1-\wb_{\max}}$.
\end{proof}

The bound on $\tranFunCo(0)$ for the predecessor-following control strategy is
one (note $\wb_{\max}=0$), which is the minimum amidst all control strategies.
For the symmetric bidirectional control we use (\ref{eq:dcgainplatoon}) to get
the steady-state gain equal to $\contNode-1$, which shows that it is unbounded in
$\numVeh$. This can be explained by the fact that all the vehicles ahead of
the vehicle $\contNode$ have to increase the distance to neighbors by one. The
steady-state gains for a fixed control node and a varying output node for several strategies are in Fig.~\ref{fig:dcgainCI}, while the gain from $\contNode$ to $\contNode$ is in Fig.~\ref{fig:dcgainII}.
Although the gain grows with $\contNode$, for a fixed $\contNode$, it does not
grow with the number $\numVeh$ of agents. 

One might also be interested in the change of the intervehicular distance
$\distPos_\obsvNode=\pos_{\obsvNode-1}-\pos_\obsvNode$ as an effect of the input $\inp_\contNode$. Then $\tranFunDist(s)=\frac{\distPos_\obsvNode(s)}{\inp_\contNode(s)}=\diagTransBlockEig_{\contNode,
\obsvNode-1}(s)-\diagTransBlockEig_{\contNode, \obsvNode}(s)$. Using
(\ref{eq:dcgainplatoon}) and (\ref{eq:pathWeight}), its
steady-state gain is  $\tranFunDist(0)=0$ for $\obsvNode \geq \contNode$ and
$\tranFunDist(0)=-\prod_{i=\obsvNode}^{\contNode-1} \wb_i$ for $\obsvNode \leq
\contNode$. This means that all the
vehicles ahead of $\contNode$ have to increase their steady-state distances (unless $\wb_i=0,\, \forall i$),
while distances of the cars behind $\contNode$ remain unchanged. In
asymmetric control with $\wb_i \leq \wbmax < 1 \, \forall i$ the change in distance will be less than one since $\prod_{i=\obsvNode}^{\contNode-1} \wb_i < 1$.

\begin{figure}
\centering
	\begin{subfigure}[b]{0.24\textwidth}
	\includegraphics[width=1\textwidth]{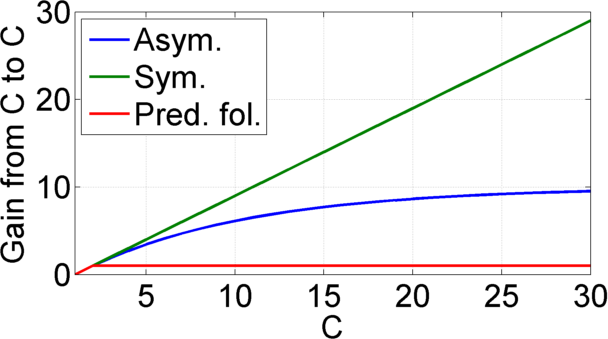}
	\caption{Varying $\contNode$ and $\obsvNode$, $\obsvNode=\contNode$.}
	\label{fig:dcgainII}
	\end{subfigure} 
	\begin{subfigure}[b]{0.24\textwidth}
	\includegraphics[width=1\textwidth]{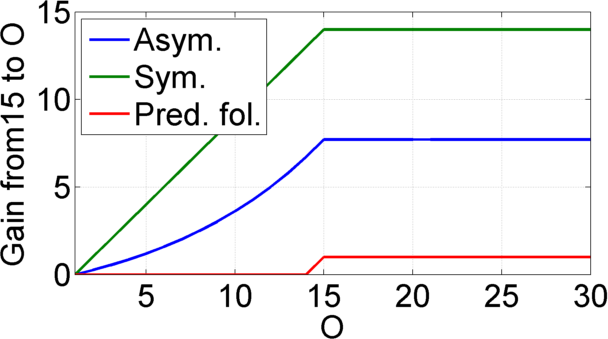}
	\caption{Fixed $\contNode=15$, varying $\obsvNode$.}
	\label{fig:dcgainCI}
	\end{subfigure}
	\caption{Steady-state gains for different choices of $\contNode$ and $\obsvNode$ and for asymmetry $\wb=0.9$.}
\end{figure}

\section{Scaling of $\hinfnorm$ norms in platoons}
In this section we investigate how the $\hinfnorm$ norm of an arbitrary transfer
function $\tranFunCo(s)$ changes when more vehicles are added ($\numVeh$ grows).
Define two types of transfer functions
	\begin{equation}
	\diagTransBlockEig_j(s) =
	\frac{\spatEig_j \vehNumCoef(s)\contNumCoef(s)}{\vehDenCoef(s)\contDenCoef(s) +
	\spatEig_j \vehNumCoef(s)\contNumCoef(s)}, \,
	\diagTransBlockEigZij(s) = \frac{\vehDenCoef(s)\contDenCoef(s) +
	\spatEigZ_i \vehNumCoef(s)\contNumCoef(s)}{\vehDenCoef(s)\contDenCoef(s) + \spatEig_j
	\vehNumCoef(s)\contNumCoef(s)}. \label{eq:Zi}
	\end{equation}

From the product (\ref{eq:tranFunGen}), we can form $\distanceCO+1$ transfer functions of type $\diagTransBlockEig_j(s)$ and $\numVeh-\distanceCO-1$ of type
$\diagTransBlockEigZ_{ij}(s)$, up to the gain. Let $\diagTransBlockEigMin(s)$ be the transfer function of the closed-loop system 
 \begin{equation}
 	\diagTransBlockEigMin(s) = \frac{\spatEigMin
 	\vehNumCoef(s)\contNumCoef(s)}{\vehDenCoef(s)\contDenCoef(s) + \spatEigMin
 	\vehNumCoef(s)\contNumCoef(s)}
 \end{equation}
with $\spatEigMin $ acting as a proportional gain ($\spatEigMin > 0$ is the lower bound on $\spatEig_i, i \geq 2$)
Similarly, for the upper bound on eigenvalues $\spatEigMax$ let
$\diagTransBlockEigMax(s)$ be the corresponding closed loop. Note that
$|\diagTransBlockEigWn(0)|=1$ due to at least one integrator in the open loop,
hence $\|\diagTransBlockEigWn(s) \|_\infty \geq 1$. The next technical Lemma is
proved in Appendix \ref{sec:pfTiZi}.

\begin{lemma}
	Let $\spatEig_j \openLoop(\jmath \omega_0)=\alpha_j + \jmath \beta_j$ for some 	frequency $\omega_0 > 0$, $\alpha_j,\beta_j \in \mathbb{R}$,
	$\jmath=\sqrt{-1}$.
	Then
	\begin{itemize}
	  \item[a)] If $|\diagTransBlockEig_{i}(\jmath \omega_0)| > 1$, then
	  $|\diagTransBlockEigWn(\jmath \omega_0)| > 1 \quad \forall \spatEig_j \geq
	  \spatEig_i$ and $\alpha_j < -1/2$.
	  \item[b)] If $|\diagTransBlockEig_{i}(\jmath \omega_0)| \leq 1$,
	  then $|\diagTransBlockEigWn(\jmath \omega_0)| \leq 1 \quad \forall \spatEig_j \leq
	  \spatEig_i$ and $\alpha_j \geq -1/2$.
	  \item[c)] 
	 $
		|\diagTransBlockEigZij(\jmath \omega_0)| \geq
		|\diagTransBlockEigZij(0)|  \, \text{ for } \{\alpha_j \leq -1 \text{ and } \spatEigZ_i \geq \spatEig_j
		\}
		$
		\item[d)]
		$
		|\diagTransBlockEigZij(\jmath \omega_0)| \geq |\diagTransBlockEigZij(0)|
		\,\text{ for }  \{-1 < \alpha_j \leq -\frac{1}{2} \text{ and } \spatEigZ_i
		\leq \spatEig_j \} 
		$
		\item[e)]
		$
		|\diagTransBlockEigZij(\jmath \omega_0)| \leq |\diagTransBlockEigZij(0)|
		\,\text{ for }
		 \{\alpha_j > -\frac{1}{2} \text{ and } \spatEigZ_i \geq \spatEig_j \}
	$
	\end{itemize}
	\label{lem:normsTiZi}
\end{lemma}

\subsection{Exponential growth}
\label{sec:exponentialScaling}
It was proven in \cite{Herman2013b} that the response of the last vehicle grows exponentially in $\numVeh$ due to the presence of a uniform nonzero lower bound on the eigenvalues. However, the analysis was done only for one transfer function in the platoon and one input---the movement of the leader. The next theorem proven in  Appendix \ref{sec:pfScalingWithDistance} extends the exponential scaling to an arbitrary transfer function in a finite platoon. The test involves only the closed-loop $\diagTransBlockEigMin(s)$ of an individual agent.
\begin{theorem}
If $\|\diagTransBlockEigMin(s)\|_\infty > 1$ and the eigenvalues of $\lapl$ are uniformly bounded from zero, then there are two real constants $0 < \xi \leq 1$ and $\zeta > 1$ depending only on $\spatEigMin, \spatEigMax$ and $\openLoop(s)$ such that  $\|\diagTransBlockCO(s)\|_\infty >\zeta^{\distanceCO} \:\:\diagTransBlockCO(0) \:\: \xi^2$. That is,  the norm $\|\tranFunCo(s)\|_\infty$ grows exponentially with the graph distance $\distanceCO$.
\label{thm:scalingWithDistance} 
\end{theorem}
The effect of the input $\inp_\contNode$ applied at the control node gets exponentially amplified with the  graph distance between $\contNode$ and $\obsvNode$. Hence, it is amplified as it propagates further from the control node even in a platoon with fixed $\numVeh$.
Figure \ref{fig:scalingWithDistance} shows scaling for a third-order model with varying asymmetry in a given range. If $\obsvNode < \contNode$, then $\tranFunCo(0)$ given in (\ref{eq:dcgainplatoon}) might decrease faster than $\zeta ^ {\distanceCO}$ grows and the norm might be less than one (Fig. \ref{fig:freqCharRearFront}). If $\contNode \leq \obsvNode$, then $\|\tranFunCo(s) \|_\infty \gg 1$  for large $\distanceCO$ (Fig. \ref{fig:freqCharFrontRear}).  In Fig. \ref{fig:hinfNorm} we show how $\|\tranFunCo(s)\|_\infty$ changes with a graph distance ---  $\contNode=3$ is kept fixed and $\obsvNode$ is varied, so that $\distanceCO$ grows with growing $\obsvNode$.

\begin{figure*}
\centering 
	\begin{subfigure}[b]{0.24\textwidth}
	\includegraphics[width=1\textwidth]{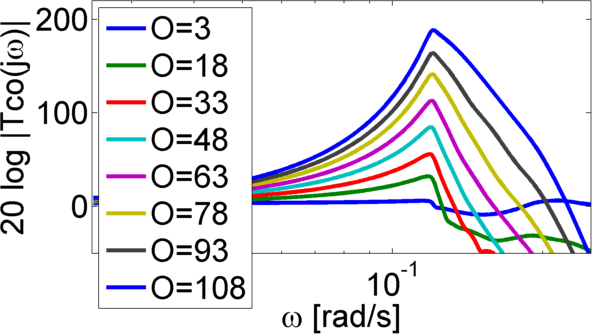}
	\caption{$\contNode=3$, $\obsvNode\geq \contNode$, $\wb_i \in [0.4, 0.6]$}
	\label{fig:freqCharFrontRear}
	\end{subfigure}   
	\begin{subfigure}[b]{0.24\textwidth}
	\includegraphics[width=1\textwidth]{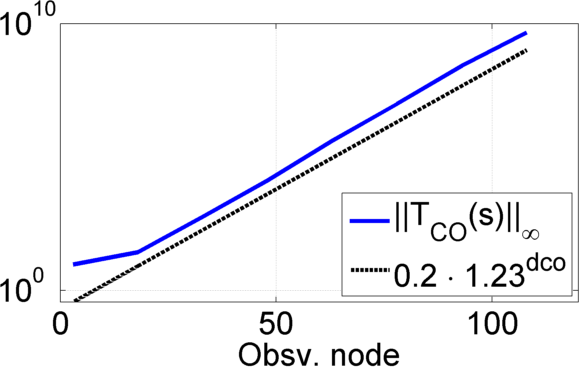}
	\caption{$\|\tranFunCo(s)\|_\infty$ for a), $\contNode=3$}
	\label{fig:hinfNorm}
	\end{subfigure}   
	\begin{subfigure}[b]{0.24\textwidth}
	\includegraphics[width=1\textwidth]{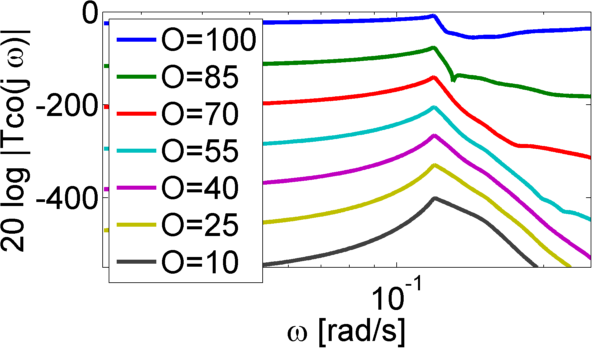}
	\caption{$\contNode=105$, $\obsvNode \leq \contNode$, $\wb_i \in [0.4, 0.6]$}
	\label{fig:freqCharRearFront}
	\end{subfigure}   
	\begin{subfigure}[b]{0.24\textwidth}
	\includegraphics[width=1\textwidth]{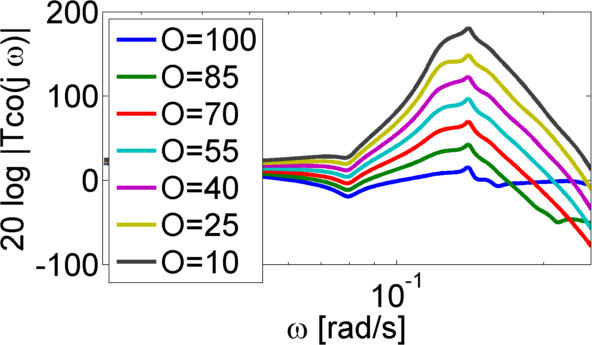}
	\caption{$\contNode=105$, $\contNode \geq \obsvNode$, $\wb_i \in [1.4, 1.6]$}
	\end{subfigure}  
	\caption{Scaling of  $|\diagTransBlockCO(\jmath \omega)|$ as a
	function of $\contNode$ kept fixed and $\obsvNode$ varying with $\numVeh=110$. The model is a PI controller $\controllerTf=\frac{s+1}{s}$ designed for a vehicle model $\vehicleTf\!=\!\frac{1}{s^2+5s}$, hence $\numInteg=2$ and the vehicle can track the leader moving with constant velocity. $\wb_i$ were randomly generated in the given range. Fig. \ref{fig:hinfNorm} shows $\|\tranFunCo(s) \|_\infty$ for the pairs $\contNode, \obsvNode$ used in a) in semilog. coordinates. It is clear that the norm scales exponentially.} 
	\label{fig:scalingWithDistance}    
\end{figure*}  

Two integrators in the open loop ($\numInteg=2$) are necessary for tracking of the 
leader moving with a constant velocity \cite[Lem. 3.1]{Yadlapalli2006}. However,
for at least two integrators in the open-loop we have
$\|\diagTransBlockEigMin(s)\|_\infty > 1$ \cite[Thm. 1]{Seiler2004a}. For
Laplacian with uniformly bounded eigenvalues this means that $\|\tranFunCo(s)
\|_\infty$ grows exponentially with the distance $\distanceCO$ and there is no
linear controller which could prevent this. Thus, we cannot have a good behavior
with a uniform bound and two integrators. The main results of \cite{Seiler2004a,
Tangerman2012, Herman2013b} are special cases of Theorem
\ref{thm:scalingWithDistance}, since asymmetric Laplacian with $\wb_i \leq
\wbmax < 1$ has uniformly bounded eigenvalues, see Lemma \ref{lem:laplProp} e).
Nevertheless, even a platoon with $\numInteg=1$ can
exhibit exponential scaling.
 \begin{figure*}
\centering
	\begin{subfigure}[b]{0.24\textwidth}
	\includegraphics[width=1\textwidth]{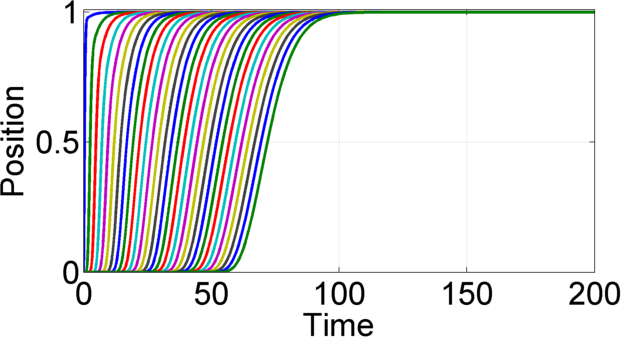}  
	\caption{Pred.fol. with $\controllerTf_1(s)$}
	\label{fig:pfC1} 
	\end{subfigure}  
	\begin{subfigure}[b]{0.24\textwidth}
	\includegraphics[width=1\textwidth]{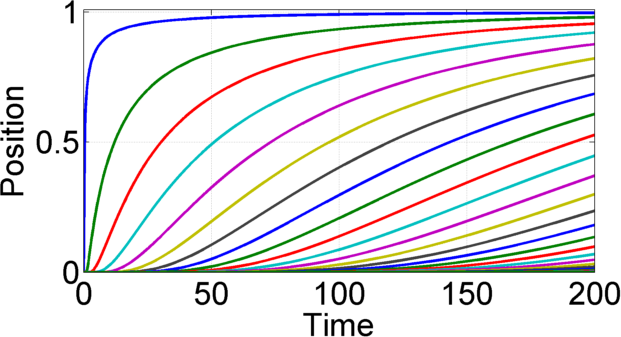}
	\caption{Asym., $\controllerTf_1(s)$, $\wb=0.9$}
	\label{fig:asymC1}
	\end{subfigure}  
	\begin{subfigure}[b]{0.24\textwidth}
	\includegraphics[width=1\textwidth]{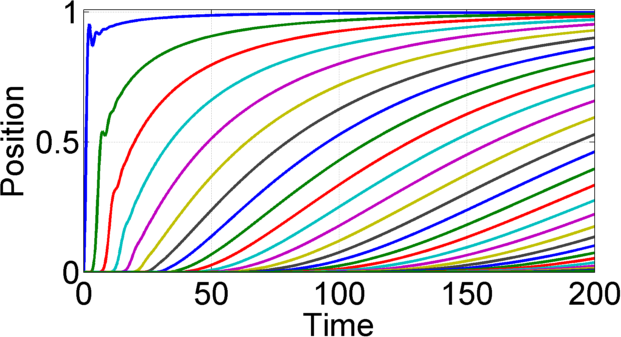}
	\caption{Asym., $\controllerTf_2(s)$, $\wb=0.9$}
	\label{fig:asymC2} 
	\end{subfigure} 
	\begin{subfigure}[b]{0.24\textwidth} 
	\includegraphics[width=1\textwidth]{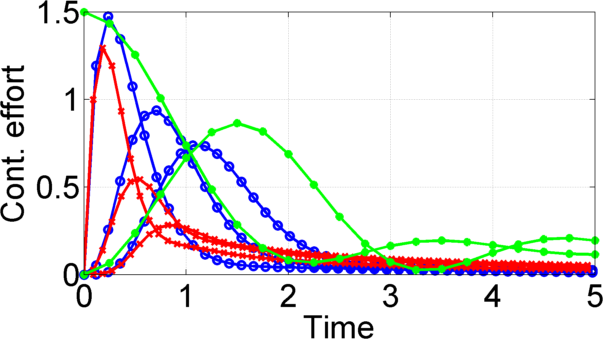}
	\caption{Control effort}
	\label{fig:contEffort}   
	\end{subfigure} 	
	\caption{Responses to leader's step in position for different architectures for $\numVeh=150$. In \ref{fig:contEffort}: blue - pred. fol., red - asym. with $\controllerTf_1(s)$, green - asym. with $\controllerTf_2(s)$ for the first three vehicles.}
	\label{fig:scalingPF} 
\end{figure*}

\subsection{Design of a string stable controller}
\label{sec:stringStabCont}
So far we have discussed situations in which the system scales badly. In this
section we provide a test for the string stability. One of the most common string stability conditions in vehicular platoons is $\left \|\frac{\pos_i(s)}{\pos_{i-1}(s)} \right\|_\infty \leq 1 \quad \forall \,i$, used e.g., in \cite{Milanes2014} (see \cite{Ploeg2014} for
other definitions). In other words, the effect of disturbance at one vehicle must be attenuated  when propagated
along the platoon. However, in a bidirectional platoon the signal can propagate
in both directions.
\begin{definition}[Bidirectional string stability]
	The bidirectional platoon is \textit{string-stable} if for an input
	$\inp_\contNode$ acting at vehicle $\contNode$ the output $\pos_\obsvNode$
	at vehicle $\obsvNode$ satisfies
	\begin{equation}
		\left \|\frac{\pos_\obsvNode(s)}{\pos_{\obsvNode-1}(s)} \right\|_\infty \leq
		1, \, \forall \, \obsvNode \geq \contNode; \,\,
		\left \|\frac{\pos_{\obsvNode-1}(s)}{\pos_{\obsvNode}(s)} \right\|_\infty
		\leq 1, \, \forall \, \obsvNode < \contNode.	
		\label{eq:stringstabbid}
	\end{equation}	
\end{definition}

We can now state a very simple sufficient condition for the bidirectional string
stability, again involving only a norm of the closed loop of an individual agent. The proof is in Appendix \ref{sec:pfTfBelowDc}.
\begin{theorem}
	If $\|\diagTransBlockEigMax(s)\|_\infty = 1$, then
	$\|\diagTransBlockCO(s)\|_\infty = |\diagTransBlockCO(0)|$ and the platoon is
	bidirectionally string-stable.
		\label{thm:bidStringStab}
\end{theorem}
The first part states that the $\hinfnorm$ norm of $\tranFunCo(s)$ equals its
steady-state gain (which is only a function of the
interconnection structure). If $\spatEigMax$ is independent of $\numVeh$, the
bidirectional string stability holds for all $\numVeh$, all $\wb_i$ and for
every $\diagTransBlockCO(s)$.

The condition $\|\diagTransBlockEigMax(s)\|_\infty = 1$ provides a simple way
how to tune a SISO controller for a vehicle model $\vehicleTf(s)$ in a platoon
of arbitrary size. To achieve $\|\diagTransBlockEigMax(s)\|_\infty = 1$, there
must be at most one integrator in the open loop. Systems with one integrator
in the open loop were used in \cite{Barooah2009a, Lin2012}, despite the fact
that they cannot track the leader's position. This is usually overcome using leader's velocity as the
reference velocity. However, this is a \emph{centralized information} and the
leader's velocity needs to be broadcast perpetually, which requires a
communications infrastructure.

\subsection{Design of a predecessor following controller} 
\label{sec:prefFolDesign}
For a platoon with uniformly bounded eigenvalues it follows from Theorem \ref{thm:scalingWithDistance} that $\| \diagTransBlockEigMin \|_\infty = 1$ is
necessary for string stability. Denote a standard closed-loop as $\diagTransBlockEig(s)=\openLoop(s)/(1+\openLoop(s))$.
\begin{lemma}
	If there is a bidirectionally-string-stable asymmetric control for a given
	$\vehicleTf(s)$, then there always exists a predecessor following controller
	($\wb=0$) achieving $\|\diagTransBlockEig(s) \|_\infty= 1$.
\end{lemma}
As an example of the closed loop, take $\diagTransBlockEig(s) =
\diagTransBlockEigMin(s)$ since $\| \diagTransBlockEigMin(s) \|_\infty=1$---the gain of the controller was just decreased to $\spatEigMin$. Since such a system
might have a slow transient response, the controller can be redesigned. 
 
The simulation results are in Fig. \ref{fig:scalingPF}. We designed two
controllers for the system  model $\vehicleTf(s)=\frac{1}{s^2+0.5s}$. The
controller $\controllerTf_1(s)=\frac{2.4s+1}{0.125s+1}$ achieves $\| \diagTransBlockEig(s)
\|_\infty = 1 $ for predecessor following (PF). In addition to that, it also has
a positive impulse response, which is very useful in platoon control. Both properties together guarantee string stability for PF in
$\mathcal{L}_\infty$-induced norm \cite{Eyre1998}. The necessary conditions for positive
response are dominant real pole and no real zero right from this pole 
\cite{Darbha2003}. The controller
$\controllerTf_2(s)=1.5$ is a simple proportional controller. A controller with
a lower gain was used in \cite{Barooah2009a}. It is apparent from Fig.
\ref{fig:scalingPF} that for the same maximal control effort, the PF achieves
the best transient response among the cases shown.

Although in general we cannot guarantee better transients of PF compared to asymmetric bidirectional control, we think that PF offers many advantages: 1)
no need for a rear-distance sensor, 2) developed theory for a closed-loop controller design (e. g.,
$\hinfnorm$ approach), 3) easier handling of heterogeneity, 4) faster convergence time for the same
maximal control effort---with the same controller the PF has a larger spectral
gap (larger $\spatEigMin$). The performance could then be compared by
simulations. Note that although the PF can have a better transient, a bidirectional
architecture might still be required, e.g., for safety reasons. Then Theorem
\ref{thm:bidStringStab} gives a condition for design.

\section{Conclusion}
We investigated asymmetric control of vehicular platoons where proportional asymmetry is used---the front spacing error is proportional to the rear spacing error. First we analyzed scaling of steady-state gain of an arbitrary
transfer function in a platoon. It was proved that it grows without bound with $\numVeh$ for a
symmetric bidirectional control scheme, while it stays bounded in a presence of
asymmetry. We proved that for more than one integrator in the open loop, the
asymmetric bidirectional control is not scalable, because the $\hinfnorm$ norm
of any transfer function grows exponentially with the graph distance. If we allow
the vehicles to know the leader's velocity (which requires
permanent communication), only one integrator in the open loop can be present.
Then we provide a simple design method for tuning the controller to achieve bidirectional string stability. In
this case also a string-stable predecessor following controller can always be
designed. This paper thus gave an overview of the achievable performance in bidirectional control with proportional asymmetry.

\appendices
\section{}
\label{sec:pfDcGain}
\begin{proof}[Proof of Theorem \ref{thm:dcgain}]
As stated in Sec. II.B., we will work with $\redLapl=[l_{ij}]$. We begin by
calculating the product in the denominator of (\ref{eq:dcGain}). The product of
all $\spatEig_i$'s equals $\det \redLapl$.
	The recursive rule to calculate the determinant of tridiagonal matrix is
	\cite[Lem. 0.9.10]{Horn1996}
$
		D_{n} = l_{n,n} D_{n-1} - l_{n, n+1} l_{n+1, n} D_{n-2}, \label{eq:recRuleDet}
$
	where $D_n$ is the determinant of the submatrix of size $n$. We begin from
	bottom right corner of $\redLapl$. Then $D_1=1$ (the
	bottom-right element) and $D_2=1$.
	Then $D_3$ can be calculated as $D_3 = (1+\wb_{N-2})D_2 - \wb_{N-2}D_1 = 1$.
	By induction, the determinant of $\redLapl$ is $
		\det \redLapl= \prod_{j=2}^\numVeh \spatEig_j = 1$ for any size of
		$\redLapl$.

	Now we calculate the product in the numerator of (\ref{eq:dcGain}). It equals the
	determinant of $\redLaplCO$.
	Suppose that $\contNode \leq \obsvNode$. If $\obsvNode < \contNode$, then the indices $\contNode$ and
	$\obsvNode$ are swapped and only the weight of the path is different.
	The matrix $\redLaplCO$ reads $\redLaplCO = \text{diag}(\lapl_1, \lapl_2)$
	with 
	
	{\small
	\begin{equation}
		\lapl_1 = \left[
						\begin{matrix}
							1+\wb_2 & -\wb_2 & 0 &  .. & 0 \\
							-1 & 1+\wb_3 & -\wb_3 & \ldots & 0 \\
							\vdots & \vdots & \vdots & \ddots & \vdots \\
							0 & \ldots & 0 & -1&  1+\wb_{\contNode-1} 
						\end{matrix} 
					\right].
	\end{equation}}
	The matrix $\lapl_2$ has the same structure as $\redLapl$, hence $\det
	\lapl_2=1$. The dimensions are $\lapl_1 \in \mathbb{R}^{(\contNode-2) \times
	(\contNode-2)}\,\text{and}\, \lapl_2 \in \mathbb{R}^{(N-\obsvNode-1) \times (N-\obsvNode-1)}$. 
	
The determinant of $\lapl_1$ of size $n\times n$ can be
recursively calculated as
		$\det \lapl_{1,n} = (1+\wb_n) \det \lapl_{1,n-1} -
		\wb_{n}\det \lapl_{1, n-2}.$
Let us start from the bottom right corner again. Then $\det
\lapl_{1,1}=1+\wb_{\contNode-1}$ and $\det \lapl_{1,2} = 1+ \wb_{\contNode-1} +
\wb_{\contNode-1} \wb_{\contNode-2}$.
The determinant 
\begin{IEEEeqnarray}{rCl}
	\det \lapl_{1,3} &=& (1+\wb_{\contNode-3})\det \lapl_{1,2} -
	\wb_{\contNode-3}\det \lapl_{1,1} \nonumber \\
&=&1+ \wb_{\contNode-1} + \wb_{\contNode-1} \wb_{\contNode-2} +
	\wb_{\contNode-1}
	\wb_{\contNode-2}\wb_{\contNode-3}.
\end{IEEEeqnarray}
The pattern is now apparent and the determinant of $\lapl_1$ is
$
	\det \lapl_1 = 1+\sum_{i=1}^{\contNode-2} \prod_{j=1}^{i}
	\wb_{\contNode-j}.
$
The sum goes from 1 to $\contNode-2$ because we excluded the leader from the
formation and the vehicle $\contNode$ is part of the path from $\contNode$
to $\obsvNode$, so $\contNode-2$ vehicles remain. Since $\det \redLaplCO =
\det \lapl_1 \det \lapl_2$, the steady state gain is then
	$\diagTransBlockCO(0) = \weightCO \frac{\det \lapl_1 \det
		\lapl_2 }{\det \redLapl} 
		= \weightCO \left(1+\sum_{i=1}^{\contNode-2}
		\prod_{j=1}^{i} \wb_{\contNode-j} \right).$
\end{proof}

\section{}
\label{sec:pfTiZi}
\begin{proof}[Proof of Lemma \ref{lem:normsTiZi}]
	\ Proof of a): The
	proof can be found as a part of the proof of \cite[Thm.
	3]{Herman2013b}. It also follows from the proof that
	$|\diagTransBlockEig_{i}(\jmath \omega_0)|>1 \Leftrightarrow \alpha <
	-1/2$.
	 
	Proof of  b) follows from a). Suppose that
	$|\diagTransBlockEigWn(\jmath \omega_0)| >1$ for $\spatEig_j < \spatEig_i$.
	Then by a) also $|\diagTransBlockEig_{i}(\jmath \omega_0)| >1$, which
	contradicts the assumption $|\diagTransBlockEig_{i}(\jmath
	\omega)|\leq 1$.
	Hence, $|\diagTransBlockEigWn(\jmath \omega_0)| \leq 1$.

	Proof of statements c)-e):	
	The transfer function $\diagTransBlockEigZij(s)$ can be written as
	$\diagTransBlockEigZij(s)\! =\! \frac{1+\spatEigZ_i
		\openLoop(s)}{1+\spatEig_j \openLoop(s)}$. 
	Its squared modulus at $\omega_0$ is using 
	$\kappa_{ij}\!=\!\frac{\spatEigZ_i}{\spatEig_j}$ given as
	\begin{IEEEeqnarray}{rCl}
	&&|\diagTransBlockEigZij(\jmath \omega_0)|^2 =
	\left|\frac{1+\kappa_{ij}(\alpha_j+\jmath \beta_j)}{1+(\alpha_j+\jmath \beta_j)}\right|^2 
	\nonumber 
	\\
		&&= \kappa_{ij}^2
	\left[1\!+\!\frac{\left(\frac{1}{\kappa_{ij}}-1\right)\left(2\alpha_j+1+\frac{1}{\kappa_{ij}}\right)}{(\alpha_j
	+ 1)^2 + \beta_j^2}\right].
	\label{eq:modulusZi} 
	\end{IEEEeqnarray}
Denote the numerator
$m_{ij}=\left(\frac{1}{\kappa_{ij}}-1\right)\left(2\alpha_j+1+\frac{1}{\kappa_{ij}}\right)$.
The square of the steady-state gain is $|\diagTransBlockEigZij(0)|^2=\kappa_{ij}^2$. 
If $m_{ij} > 0$, then $|\diagTransBlockEigZij(\jmath \omega_0)|^2 > |\diagTransBlockEigZij(0)|^2=
\kappa_{ij}^2$ since $(\alpha_j
	+ 1)^2 + \beta_j^2> 0$.
If $m_{ij}\leq 0$, then $|\diagTransBlockEigZij(\jmath \omega_0)|^2 \leq
|\diagTransBlockEigZij(0)|^2$.
	Let us analyze the statements c)-e).
	
	c) If $\,\alpha_j \leq -1$ and
	$\spatEigZ_i \geq \spatEig_j$, then $\left(\frac{1}{\kappa_{ij}}-1\right) \leq
	0$ and also $\left(2\alpha_j+1+\frac{1}{\kappa_{ij}}\right) \leq 0$, hence $m_{ij} \geq 0$
	which proves the statement c).
d) If $-1 < \, \alpha_j \leq -\frac{1}{2} \text{ and } \spatEigZ_i \leq
		\spatEig_j$, so $\kappa_{ij} \leq 1$, then
		$\left(\frac{1}{\kappa_{ij}}-1\right) \geq 0$ and also
		$\left(2\alpha_j+1+\frac{1}{\kappa_{ij}}\right) \geq 0$,  $m_{ij} > 0$
		and d) is proved.	
	e) If $\, \alpha_j\!>\!-\frac{1}{2} \text{ and } \spatEigZ_i \geq
	\spatEig_j$, then $\left(\frac{1}{\kappa_{ij}}-1\right) \leq 0$ and
	$\left(2\alpha_j+1+\frac{1}{\kappa_{ij}}\right)\!\geq\!0$, hence $m_{ij} \!\leq\! 0$ and e) is proved.
\end{proof}

\section{}
\label{sec:pfScalingWithDistance}
\begin{proof}[Proof of Theorem \ref{thm:scalingWithDistance}]
In the proof we work with reduced Laplacian $\redLapl$. Let $\omega_0$ be
a frequency at which $|\diagTransBlockEigMin(\jmath \omega_0)|>1$. The key idea
 is to form $\diagTransBlockEigWn(s)$ and $\diagTransBlockEigZij(s)$
from (\ref{eq:tranFunGen}) as follows:
\begin{enumerate}
  \item Take each term $\vehDenCoef(s) \contDenCoef(s) +
\spatEig_j \vehNumCoef(s)\contNumCoef(s)$ from the denominator of
(\ref{eq:tranFunGen}). Let $\alpha_j+\jmath \beta_j =
\spatEig_j \openLoop(\jmath \omega_0)$. Since $|\diagTransBlockEigMin(\jmath
\omega_0)| > 1$, from Lemma \ref{lem:normsTiZi} a) we
know that $\alpha_j < -\frac{1}{2}$. 
\item If $\alpha_j \leq -1$, then find $\spatEigZ_i$
	such that $\spatEigZ_i \geq \spatEig_j$. Form
	$\diagTransBlockEigZij(s)$ using such $\spatEigZ_i$ and $\spatEig_j$. Then by
	c) in Lemma \ref{lem:normsTiZi} for such $\diagTransBlockEigZij(s)$ holds
	$|\diagTransBlockEigZij(\jmath \omega_0)| \geq |\diagTransBlockEigZij(0)|$.
	\item If $-1 < \alpha_j \leq -\frac{1}{2}$, then find $\spatEigZ_i$
	such that $\spatEigZ_i \leq \spatEig_j$. Form $\diagTransBlockEigZij(s)$ using
	these $\spatEigZ_i$ and $\spatEig_j$. Then by Lemma \ref{lem:normsTiZi} d)
	$|\diagTransBlockEigZij(\jmath \omega_0)| \geq |\diagTransBlockEigZij(0)|$.
	\item Form as much $\diagTransBlockEigZij(s)$'s as possible using the
	steps 2) and 3). Use $(\distanceCO+1)$ remaining terms $\vehDenCoef(s)
	\contDenCoef(s) + \spatEig_j \vehNumCoef(s)\contNumCoef(s)$ to form $\diagTransBlockEigWn(s)$.
\end{enumerate}
	
	\begin{figure}
\centering
	\includegraphics[width=0.28\textwidth]{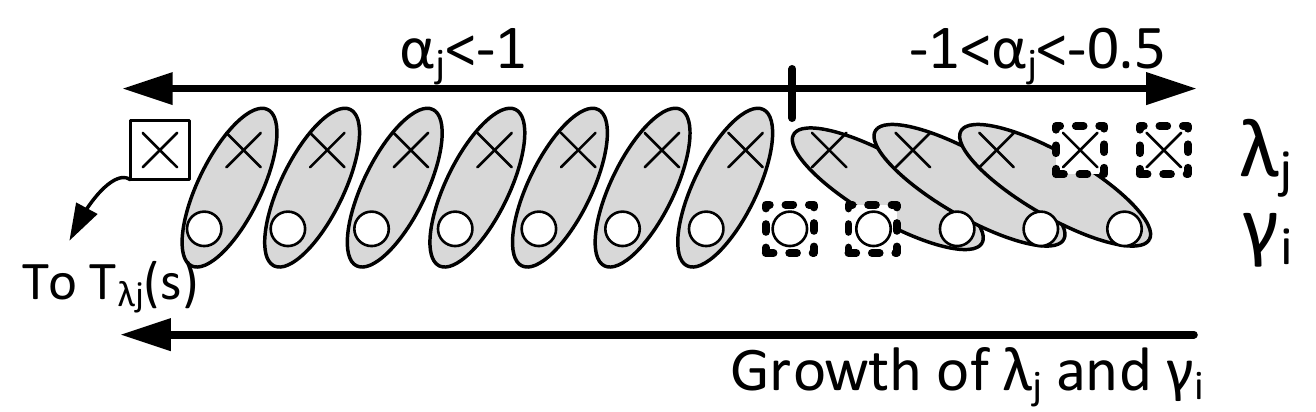}
	\caption{Matching of $\spatEig_j$ and $\spatEigZ_{i}$ to form
	$\diagTransBlockEigZij(s)$.
	Dashed pairs are the two $\diagTransBlockEigZij(s)$ for which
	$|\diagTransBlockEigZij(\jmath \omega_0)|>1$ is not guaranteed.}
	\label{fig:poleZeroMatching}
\end{figure} 
Lemma
\ref{lem:laplProp} f) allows us to find $(\numVeh-\distanceCO-3)$
$\diagTransBlockEigZij(s)$'s to satisfy either c) or d) in Lemma
\ref{lem:normsTiZi} --- we pair $\spatEigZ_i$ with $\spatEig_{i+2}$ for
$\alpha_j \leq -1$ and $\spatEigZ_i$ with $\spatEig_i$ for $-1 < \alpha_j \leq
0.5$ (see Fig. \ref{fig:poleZeroMatching}).
These $\diagTransBlockEigZij(s)$'s all have gain greater than one at $\omega_0$. The remaining two $\diagTransBlockEigZij(s)$'s might
have gain less than one. Since $\spatEig_j$ and $\spatEigZ_i$ are bounded, there
is a lower bound $\xi$ such that $|\diagTransBlockEigZij(\jmath \omega)|\geq
\xi$ for these two.

The transfer function $\tranFunCo(s)$ given in
(\ref{eq:tranFunGen}) is using such $\diagTransBlockEigWn$'s and
$\diagTransBlockEigZij$'s written as
	\begin{equation}
		\diagTransBlockCO(s) =
		\weightCO \,\, \displaystyle{\prod_{i=2, j \in
		\setZi}^{\numVeh-\distanceCO-1}}\,\, \diagTransBlockEigZij(s)	
		\prod_{j=2,j \notin \setZi}^{\numVeh} \frac{1}{\spatEig_j}
		\prod_{j=2, j \notin \setZi}^{\numVeh}\diagTransBlockEigWn(s).
		\label{eq:tranFunGen2}
	\end{equation}
The set $\mathcal{J}$ is the set of $\spatEig_j$ used to form some of
$\diagTransBlockEigZij$'s. The terms $\weightCO \prod_{j=2, j \notin
\setZi}^{\numVeh} \frac{1}{\spatEig_j}$ and steady-state gain of
$\diagTransBlockEigZij(0)$ do not affect the shape of the magnitude
frequency response, only its value.
	
Since $\|\diagTransBlockEigMin(s)\|_\infty > 1$, it follows from
a) in Lemma \ref{lem:normsTiZi}  that for all transfer functions
$\diagTransBlockEigWn(s)$ we have $|\diagTransBlockEigWn(\jmath \omega_0)| > 1$. 
Due to the lower and upper bounds on eigenvalues, there is a minimum $\zeta>1$
of modulus frequency response $|\diagTransBlockEigWn(\jmath
\omega_0)|$, attained for some $\spatEig_j$ with $\spatEigMin \leq \spatEig_j
\leq \spatEigMax$. Then we get the lower bound
on the modulus of product of $\diagTransBlockEigWn(s)$ in (\ref{eq:tranFunGen2})
as
$
	\prod_{j=2, j \notin \setZi}^{\numVeh} |\diagTransBlockEigWn(\jmath \omega_0)|
	\geq \zeta^{\distanceCO+1}.
$
Clearly, this part of (\ref{eq:tranFunGen2}) scales exponentially with
$\distanceCO$.

All but two blocks $\diagTransBlockEigZij(s)$ amplify at $\omega_0$, so 
$	\prod_{i=1}^{\numVeh-\distanceCO-1} |\diagTransBlockEigZij(\jmath \omega_0)|
	\geq \xi^2$ (excluding the
steady-state gain)
and the norm of $\diagTransBlockCO(s)$ is from
(\ref{eq:tranFunGen2})
$
	\|\diagTransBlockCO(s)\|_\infty \geq \xi^2 \,
		\diagTransBlockCO(0) \, \zeta^{\distanceCO}.
$
\end{proof}
\section{}
\label{sec:pfTfBelowDc}
\begin{proof}[Proof of Theorem \ref{thm:bidStringStab}]
	First we prove that if $\|\diagTransBlockEigMax(s)\|_\infty = 1$, then
	$\|\diagTransBlockCO(s)\|_\infty = |\diagTransBlockCO(0)|$. As in the proof of
	Theorem \ref{thm:scalingWithDistance}, we will form $\diagTransBlockEigZij$'s and $\diagTransBlockEigWn$'s in a suitable way.	
Let $\alpha_j+\jmath \beta_j = \spatEig_j \openLoop(\jmath \omega_0)$ at
some frequency $\omega_0$. Since $\|\diagTransBlockEigMax(s)
\|_\infty = 1$, it follows from Lemma \ref{lem:normsTiZi} b)
 that $|\diagTransBlockEigWn(\jmath \omega_0)| \leq 1\, \forall \omega_0,\,
 \forall \spatEig_j \leq \spatEigMax$ and $\alpha_j \geq
 -\frac{1}{2}, \, \forall \omega_0$.
	 
Using Lemma \ref{lem:laplProp} f) we can pair all $\spatEigZ_i$ with
unique $\spatEig_j$ such that $\spatEigZ_i \geq \spatEig_j$ to form
$\diagTransBlockEigZij(s)$. Then e) in Lemma \ref{lem:normsTiZi} implies that
$|\diagTransBlockEigZij(\jmath \omega_0)| \leq |\diagTransBlockEigZij(0)|$ for all $i,j$.
Since $\alpha_j \geq -\frac{1}{2}$ for all $\omega_0$, we have that
$\|\diagTransBlockEigZij(s)\|_\infty = |\diagTransBlockEigZij(0)|$ for all
pairs $\spatEigZ_i \geq \spatEig_j$. All remaining terms
$\diagTransBlockEigWn(s)$ in (\ref{eq:tranFunGen2}) by Lemma
\ref{lem:normsTiZi}b) satisfy $|\diagTransBlockEigWn(\jmath
\omega_0)|\leq 1$ for all $\omega_0$.
Hence, all transfer functions in the product (\ref{eq:tranFunGen2}) have their
norm norm less than or equal to one and $\|\diagTransBlockCO(s) \|_\infty =
|\diagTransBlockCO(0)|$ .

	Now let us go back to bidirectional string stability. Consider $\obsvNode \geq
	\contNode$ and let $\inp_\contNode$ be the input at the control node.
	Then the first transfer function in (\ref{eq:stringstabbid}) can be written as
	\begin{multline}
		\frac{\pos_\obsvNode(s)}{\pos_{\obsvNode-1}(s)} = \frac{\inp_\contNode(s)
		\diagTransBlockEig_{\contNode, \obsvNode}(s)}{\inp_\contNode(s)
		\diagTransBlockEig_{\contNode, \obsvNode-1}(s)} = \frac{
		\diagTransBlockEig_{\contNode, \obsvNode}(s)}{
		\diagTransBlockEig_{\contNode, \obsvNode-1}(s)} 
\\=
		\frac{\vehNumCoef(s)
		\contNumCoef(s)\prod_{j=1}^{N-\distanceCO-1}\vehDenCoef(s)\contDenCoef(s)
		+ \spatEigZ_{j,
		\obsvNode}
		\,\vehNumCoef(s)\contNumCoef(s)}{\prod_{j=1}^{N-\distanceCO}\vehDenCoef(s)\contDenCoef(s)
		+ \spatEigZ_{j, \obsvNode-1}\,\vehNumCoef(s)\contNumCoef(s)}.
		\label{eq:prodZerosXn}
	\end{multline}
	 Let $\redLaplCO_{\obsvNode-1}$ and
	$\redLaplCO_{\obsvNode}$ be the submatrices of $\lapl$ corresponding to the paths from $\contNode$ to $\obsvNode-1$ and from $\contNode$ to $\obsvNode$, respectively. 
Their eigenvalues are $\spatEigZ_{j, \obsvNode-1}$ and $\spatEigZ_{j,
\obsvNode}$, respectively. Beacuse of the fact that $\redLaplCO_{\obsvNode}$ is
a submatrix of $\redLaplCO_{\obsvNode-1}$, the eigenvalues of
$\redLaplCO_{\obsvNode-1}$ and $\redLaplCO_{\obsvNode}$ must interlace in a
sense of f) in Lemma \ref{lem:laplProp}.
We can pair $\spatEigZ_{j, \obsvNode-1}$ and $\spatEigZ_{j, \obsvNode}$ by Lemma
\ref{lem:laplProp} f) such that $\spatEigZ_{j, \obsvNode-1} \leq \spatEigZ_{j,
\obsvNode}$ and form $\diagTransBlockEigZij(s)$ as above. Then,
	$
		\left \|\frac{\vehDenCoef(s)\contDenCoef(s) +
		\spatEigZ_{j, \obsvNode}\vehNumCoef(s)\contNumCoef(s)}{\vehDenCoef(s)\contDenCoef(s) + \spatEigZ_{j,
		\obsvNode-1}\vehNumCoef(s)\contNumCoef(s)} \right\|_\infty \leq 1 \, \forall
		j.
	$
	Only one term in (\ref{eq:prodZerosXn}) with a
	form $\frac{\vehNumCoef(s)\contNumCoef(s)}{\vehDenCoef(s)\contDenCoef(s) +
	\spatEigZ_{i, \obsvNode-1}\vehNumCoef(s)\contNumCoef(s)}$ remains. Its
	$\hinfnorm$ norm is less than or equal to one by b) in Lemma
	\ref{lem:normsTiZi}.
	The steady-state gain of $\frac{\pos_\obsvNode(s)}{\pos_{\obsvNode-1}(s)}$ is one, since by Theorem \ref{thm:dcgain} the steady-state gain is identical for all the vehicles behind the
	control node.
	Hence, $\|\frac{\pos_\obsvNode(s)}{\pos_{\obsvNode-1}(s)}\|_\infty \leq 1$ for
	$\contNode \leq \obsvNode$.
	
	The other direction ($\contNode \geq \obsvNode$) has the ratio of outputs 
	with the same structure as (\ref{eq:prodZerosXn}), the only difference is its
	steady-state gain. It follows from (\ref{eq:dcgainplatoon}) that the steady-state gain is
	\begin{IEEEeqnarray}{rCl}
		&&\frac{\diagTransBlockEig_{\contNode, \obsvNode-1}(0)}{
		\diagTransBlockEig_{\contNode, \obsvNode}(0)} \!=\!  \wb_{\obsvNode-1} \frac{\left(1+\sum_{i=1}^{\obsvNode-3}
		\prod_{j=1}^{i} \wb_{\obsvNode-j-1} \right)}{\left(1+\sum_{i=1}^{\obsvNode-2}
		\prod_{j=1}^{i} \wb_{\obsvNode-j} \right)} 
		 < 1.
	\end{IEEEeqnarray}
	Since the norm $\|{\pos_{\obsvNode-1}(s)}/{\pos_{\obsvNode}(s)}\|_\infty$ is
	at most 1, bidirectional string stability was proved.
\end{proof}
\bibliographystyle{IEEEtran}  
\bibliography{Papers-MistunedControlDrawbacks}

\begin{thebibliography}{10}
\providecommand{\url}[1]{#1}
\csname url@samestyle\endcsname
\providecommand{\newblock}{\relax}
\providecommand{\bibinfo}[2]{#2}
\providecommand{\BIBentrySTDinterwordspacing}{\spaceskip=0pt\relax}
\providecommand{\BIBentryALTinterwordstretchfactor}{4}
\providecommand{\BIBentryALTinterwordspacing}{\spaceskip=\fontdimen2\font plus
\BIBentryALTinterwordstretchfactor\fontdimen3\font minus
  \fontdimen4\font\relax}
\providecommand{\BIBforeignlanguage}[2]{{%
\expandafter\ifx\csname l@#1\endcsname\relax
\typeout{** WARNING: IEEEtran.bst: No hyphenation pattern has been}%
\typeout{** loaded for the language `#1'. Using the pattern for}%
\typeout{** the default language instead.}%
\else
\language=\csname l@#1\endcsname
\fi
#2}}
\providecommand{\BIBdecl}{\relax}
\BIBdecl

\bibitem{Milanes2014}
V.~Milan\'{e}s, S.~Shladover, J.~Spring, C.~Nowakowski, H.~Kawazoe, and
  M.~Nakamura, ``{Cooperative adaptive cruise control in real traffic
  situations},'' \emph{IEEE Transactions on Inteligent Transportation Systems},
  vol.~15, no.~1, pp. 296--305, 2014.

\bibitem{Coelingh2012}
E.~Coelingh and S.~Solyom, ``{All aboard the robotic road train},'' \emph{IEEE
  Spectrum}, vol.~49, no.~11, pp. 34--39, 2012.

\bibitem{Naus2010}
G.~Naus, R.~Vugts, J.~Ploeg, R.~van~de Molengraft, and M.~Steinbuch,
  ``{Cooperative adaptive cruise control, design and experiments},'' in
  \emph{American Control Conference 2010}, no.~1, 2010, pp. 6145--6150.

\bibitem{Peters2013}
A.~Peters, R.~H. Middleton, and O.~Mason, ``{Leader tracking in homogeneous
  vehicle platoons with broadcast delays},'' \emph{Automatica}, no. 1992, Oct.
  2013.

\bibitem{Hao2012}
H.~Hao and P.~Barooah, ``{On achieving size-independent stability margin of
  vehicular lattice formations with distributed control},'' \emph{IEEE
  Transactions on Automatic Control}, vol.~57, no.~10, pp. 2688--2694, 2012.

\bibitem{Barooah2009a}
P.~Barooah, P.~G.~P. Mehta, and J.~J.~P. Hespanha, ``{Mistuning-Based Control
  Design to Improve Closed-Loop Stability Margin of Vehicular Platoons},''
  \emph{IEEE Transactions on Automatic Control}, vol.~54, no.~9, pp.
  2100--2113, Sep. 2009.

\bibitem{Middleton2010}
R.~H. Middleton and J.~H. Braslavsky, ``{String Instability in Classes of
  Linear Time Invariant Formation Control With Limited Communication Range},''
  \emph{IEEE Transactions on Automatic Control}, vol.~55, no.~7, pp.
  1519--1530, 2010.

\bibitem{Seiler2004a}
P.~Seiler, A.~Pant, and K.~Hedrick, ``{Disturbance propagation in vehicle
  strings},'' \emph{IEEE Transactions on Automatic Control}, vol.~49, no.~10,
  pp. 1835--1841, 2004.

\bibitem{Hao2012c}
H.~Hao, H.~Yin, and Z.~Kan, ``{On the robustness of large 1-D network of double
  integrator agents},'' in \emph{American Control Conference (ACC), 2012,
  Montreal, Canada}, 2012, pp. 6059--6064.

\bibitem{Cantos2014a}
C.~E. Cantos and J.~J.~P. Veerman, ``{Transients in the Synchronization of
  Oscillator Arrays},'' \emph{Arxiv preprint:1308.4919v4}, pp. 1--11, 2014.

\bibitem{Tangerman2012}
F.~Tangerman, J.~Veerman, and B.~Stosic, ``{Asymmetric decentralized flocks},''
  \emph{IEEE Transactions on Automatic Control}, vol.~57, no.~11, pp.
  2844--2853, 2012.

\bibitem{Herman2013b}
I.~Herman, D.~Martinec, Z.~Hurak, and M.~Sebek, ``{Nonzero Bound on Fiedler
  Eigenvalue Causes Exponential Growth of H-Infinity Norm of Vehicular
  Platoon},'' \emph{IEEE Transactions on Automatic Control}, vol.~60, no.~8,
  pp. 2248--2253, 2015.

\bibitem{Veerman2007}
J.~Veerman, B.~Stosic, and A.~Olvera, ``{Spatial instabilities and size
  limitations of flocks},'' \emph{Networks and Heterogeneous Media}, vol.~2,
  no.~4, pp. 1--14, 2007.

\bibitem{Hao2012b}
H.~Hao and P.~Barooah, ``{Stability and robustness of large platoons of
  vehicles with double-integrator models and nearest neighbor interaction},''
  \emph{International Journal of Robust and Nonlinear Control}, vol.~23, pp.
  2097--2122, 2012.

\bibitem{Martinec2014b}
D.~Martinec, I.~Herman, Z.~Hur\'{a}k, and M.~\v{S}ebek, ``{Wave-absorbing
  vehicular platoon controller},'' \emph{European Journal of Control}, vol.~20,
  pp. 234--248, 2014.

\bibitem{Bamieh2012}
B.~Bamieh, M.~R. Jovanovi\'{c}, P.~Mitra, and S.~Patterson, ``{Coherence in
  Large-Scale Networks: Dimension-Dependent Limitations of Local Feedback},''
  \emph{IEEE Transactions on Automatic Control}, vol.~57, no.~9, pp.
  2235--2249, 2012.

\bibitem{Fallat2011}
S.~Fallat and C.~Johnson, \emph{{Totally Nonnegative Matrices}}.\hskip 1em plus
  0.5em minus 0.4em\relax Princeton University Press, 2011.

\bibitem{Herman2014a}
I.~Herman, D.~Martinec, and M.~Sebek, ``{Zeros of Transfer Functions in Network
  Control with Higher-Order Dynamics},'' in \emph{19th IFAC World Congress},
  Cape Town, South Africa, 2014, pp. 9177--9182.

\bibitem{Fax2004a}
J.~A. Fax and R.~Murray, ``{Information flow and cooperative control of vehicle
  formations},'' \emph{IEEE Transactions on Automatic Control}, vol.~49, no.~9,
  pp. 1465--1476, Sep. 2004.

\bibitem{Wieland2011}
P.~Wieland, R.~Sepulchre, and F.~Allg\"{o}wer, ``{An internal model principle
  is necessary and sufficient for linear output synchronization},''
  \emph{Automatica}, vol.~47, no.~5, pp. 1068--1074, May 2011.

\bibitem{Yadlapalli2006}
S.~K. Yadlapalli, S.~Darbha, and K.~R. Rajagopal, ``{Information flow and its
  relation to the stability of the motion of vehicles in a rigid formation},''
  \emph{IEEE Transactions on Automatic Control}, vol.~51, no.~8, pp.
  1315--1319, 2006.

\bibitem{Ploeg2014}
J.~Ploeg, N.~van~de Wouw, and H.~Nijmeijer, ``{Lp String Stability of Cascaded
  Systems: Application to Vehicle Platooning},'' \emph{IEEE Transactions on
  Control Systems Technology}, vol.~22, no.~2, pp. 786--793, 2014.

\bibitem{Lin2012}
F.~Lin, M.~Fardad, and M.~R. Jovanovi\'{c}, ``{Optimal control of vehicular
  formations with nearest neighbor interactions},'' \emph{IEEE Transactions on
  Automatic Control}, vol.~57, no.~9, pp. 2203--2218, 2012.

\bibitem{Eyre1998}
J.~Eyre, D.~Yanakiev, and I.~Kanellakopoulos, ``{A Simplified Framework for
  String Stability Analysis of Automated Vehicles},'' \emph{Vehicle System
  Dynamics}, 1998.

\bibitem{Darbha2003}
S.~Darbha, ``{On the synthesis of controllers for continuous time LTI systems
  that achieve a non-negative impulse response},'' \emph{Automatica}, vol.~39,
  no.~1, pp. 159--165, 2003.

\bibitem{Horn1996}
R.~Horn and C.~Johnson, \emph{{Matrix analysis}}.\hskip 1em plus 0.5em minus
  0.4em\relax Cambridge: Cambridge University Press, 1990.

\end{thebibliography}

\end{document}